\documentclass[pdftex,hidelinks,
final
]{dmtcs-episciences}

\usepackage[utf8]{inputenc}
\usepackage{subfigure}

\usepackage[round]{natbib}
\usepackage{xcolor}
\usepackage{comment}
\usepackage{amsfonts}
\usepackage{mathtools}
\usepackage{amsthm}
\usepackage{amsmath, amssymb, bm}
\usepackage[linesnumbered,ruled,noend]{algorithm2e}
\usepackage[title,toc,titletoc,header]{appendix}
\DeclareMathOperator*{\argmax}{arg\,max}

\newcommand{\then}{\Longrightarrow}
\usepackage[symbol]{footmisc}
\newtheorem{definition}{Definition}
\newtheorem{lemma}{Lemma}
\newtheorem{theorem}{Theorem}
\newtheorem{corollary}{Corollary}

\author[I. Sigalas et al.]{Ioannis Sigalas \and
    Nikolaos Lazaropoulos \and
    Ioannis Lamprou\\ \and
	Ioannis Vaxevanakis\and
	Vassilis Zissimopoulos}

\title{Domination and Coverage Problems under Vulnerability Constraints}
\affiliation{Department of Informatics and Telecommunications, National and Kapodistrian University of Athens, Greece}
\begin{document}
\maketitle
\begin{abstract}
 In various domination and coverage problems, certain vertices or edges should not be dominated/covered and are designated as vulnerable. Motivated by this, we define the \textit{$k$-Vertex Maximum Domination Ratio with Vulnerable Vertices} ($k\textit{-}Max \ \mathit{DRVV}$) problem, which extends the budgeted dominating set problem to include vulnerability constraints. We propose an approximation algorithm based on an unbudgeted variant of $k\textit{-}Max \ \mathit{DRVV}$, termed the \textit{Maximum Domination Ratio with Vulnerable Vertices} ($\mathit{DRVV}$) problem. For bounded-degree graphs of order $n$, our algorithm provides an $O(k/n)$-approximation for the $k\textit{-}Max \ \mathit{DRVV}$ problem. We introduce the Dominating Set with Vulnerable Vertices ($\mathit{DSV}$) problem, reduce it to the Red\textit{-}Blue Set Cover problem, and derive a $2\sqrt{|V|\cdot(H(\Delta_{N})-\frac{1}{2}})$-approximation algorithm, where $|V|$ is the order of the graph, $\Delta_N$ is the maximum degree among non-vulnerable vertices and $H$ is the harmonic function. 
Finally, we examine the \textit{Vertex Cover with Vulnerable Edges} ($\mathit{VCVE}$) problem, which can be naturally expressed as a special case of the Red\textit{-}Blue Set Cover problem. 
We present a polynomial-time $2$-approximation algorithm for the $VCVE$ problem, achieving the best possible ratio.
\end{abstract}

\section{Introduction}
\label{sec:intro} 

Domination and Coverage problems have been studied extensively in the literature~\cite{haynes2013fundamentals,haynes2020topics,Vazirani:2001}. In this work, we focus on variants of these problems with constraints regarding vulnerable vertices or edges. These formulations are closely connected to the classical Red\textit{-}Blue Set Cover problem \cite{carr1999,PELEG200755} and are motivated by applications in influence maximization, where the diffusion cascade model is typically used to capture how ideas or behaviors spread in networks while respecting vulnerable elements. In contrast to cascade-based formulations, our perspective abstracts the diffusion process into structural graph properties, aiming to maximize the number of dominated vertices or covered edges while controlling the number of vulnerable ones.

Influence maximization aims to select a set of seed vertices whose activation triggers the largest spread of ideas or behaviors throughout the network. There has been significant research \cite{pasumarthi2015near,wen2018maximizing} focusing on optimizing the influence on a specific target audience. Socially responsible influence maximization \cite{CHEN202184} further emphasizes minimizing influence on vulnerable vertices. A common metric to evaluate this trade-off is the Additively Smoothed Ratio ($ASR$), the fraction of expected influenced non-vulnerable vertices over vulnerable ones. We, on the other hand, focus on combinatorial formulations that represent vulnerability constraints directly on the graph.

Building on this perspective, we define the $k$-Vertex Maximum Domination Ratio with Vulnerable Vertices ($k\textit{-}Max \ \mathit{DRVV}$), which maximizes the ratio of dominated non-vulnerable over vulnerable vertices.  This formulation offers a clear combinatorial abstraction of socially responsible influence maximization and provides a framework in which vulnerability constraints can be studied through domination and coverage problems. Vertex Cover with vulnerable edges ($\mathit{VCVE}$) is an example of coverage problem with vulnerable edges. 

Beyond applications in social networks, the concept of Red\textit{-}Blue vertices (or vulnerable and non-vulnerable vertices in our formulation) extends to a variety of other domains. For instance, data mining provides an important motivation, as illustrated by \cite{carr1999} in the context of detecting fraudulent claims within extensive datasets of legitimate Medicare records. Furthermore, the proposed framework may contribute to advancements in machine learning by enabling the prioritization of salient information and facilitating the enrichment of training datasets with data of higher relevance.

\subsection{Related Work}
\label{sec:relatedWork} 
The Dominating Set problem is well studied, and numerous variations have been introduced. One such budgeted variant is the k-Vertex Maximum Domination ($k\textit{-}Max \ VD$) problem, defined as follows: Given a graph $G$, find a set with at most $k$ vertices that maximizes the number of dominated vertices. \cite{miyano2011maximum} provide a $(1-\frac{1}{e})$-approximation algorithm. In the connected setting, the analogous problem is the budgeted connected dominating set problem ($BCDS$): Given a graph $G$, find a set with at most $k$ vertices that induce a connected subgraph and maximize the number of dominated vertices. \cite{khuller2013} provide a $\frac{1}{13}(1-\frac{1}{e})$-approximation algorithm, which was gradually improved up to  $\frac{1}{6}(1-\frac{1}{e})$ \cite{IBCD-IWOCA,AnalyzingtheOptimal,LAMPROU20211,10631020}. 

In the direction of Influence Maximization,  \cite{CHEN202184} deal with the Independent Cascade Model under the $ASR$ metric. Given a graph $G=(V, E)$  where vertices are partitioned into non-vulnerable $V_N$ and vulnerable $V_L$, the $ASR$ metric of a subset of vertices $S \subseteq V_N$ is $ASR(S,c) =\frac{\sigma_N(S) + c}{\sigma_L(S) + c}$. The set $S$ is called a seed-set and $\sigma_N(S)$ and $\sigma_L(S)$  are the expected number of influenced non-vulnerable vertices and the expected number of vulnerable vertices, respectively. A positive constant $c>0$ is introduced primarily to allow $ASR$ to be applied even when the seed set contains no vulnerable vertices. The choice of $c$ influences the outcome, as discussed in \cite{CHEN202184}; specifically, increasing the value of $c$ results in a larger output seed set. $c$ is part of the input of any specific instance. The \textit{ASR-Maximizing Seed-set problem} aims to identify a seed set containing no more than $k$ vertices that maximizes the $ASR$ metric. Note that $ASR$ is a non-monotone, non-submodular, and non-supermodular function. The authors propose an algorithm for the problem, that guarantees $\mathbb{E}[ASR(S,c)] \geq \frac{\sigma_L(S^*)+c}{|V|+c}(1-\frac{1}{e})ASR(S^*,c)$, where $S^*$ is an optimal solution of size at most $k$ with respect to $ASR$, and $\mathbb{E}[ASR(S,c)]$ denotes the expectation of $ASR$ over every solution $S$ constructed by their algorithm.

The Red\textit{-}Blue Set Cover ($RBSC$) problem \cite{carr1999} is a natural generalization of the Set Cover problem with vulnerability constraints. In $RBSC$, the ground set is partitioned into “blue” elements, which must be covered, and “red” elements, whose coverage should be minimized. Formally, given a universe $U = R \cup B$ (with $R \cap B = \emptyset$), the goal is to select a subcollection of sets that covers all blue elements while covering as few red elements as possible. In \cite{carr1999,Elkin2000}, it is shown that it is quasi-NP-hard to approximate the Red-Blue Set Cover problem with ratio $2^{\log^\varepsilon|U|}$, for any $0<\varepsilon<1$, and \cite{PELEG200755} provided a $2\sqrt{|U|\cdot \log|B|}$-approximation algorithm for $RBSC$.

\subsection{Contribution}\label{sec:contribution} 

This work investigates domination and coverage problems involving vulnerable vertices or edges. We define the \textit{k-Vertex Maximum Domination Ratio with Vulnerable Vertices} ($k\textit{-}Max \ \mathit{DRVV}$) problem. In order to provide an approximation algorithm for $k\textit{-}Max \ \mathit{DRVV}$, we define a version without the $k$ restriction, named the \textit{Maximum Domination Ratio with Vulnerable Vertices} ($\mathit{DRVV}$) problem. We prove that $\mathit{DRVV}$ is NP-hard by reducing its decision version from a special case of $3$-$\mathcal{SAT}$. This implies hardness
of the $k\textit{-}Max \ \mathit{DRVV}$. Next, we define the \textit{Maximum Domination with a Bounded Number of Vulnerable Vertices} ($\mathit{DVV}$) problem and provide the first $\frac{1}{\Delta+1} \cdot (1 - e^{-\frac{1}{\Delta}})$-approximation ratio, where $\Delta$ is the maximum degree of the graph. Using $\mathit{DVV}$, we propose an algorithm for $\mathit{DRVV}$ and prove that it achieves the same approximation ratio as $\mathit{DVV}$. Then, we provide an approximation algorithm for $k\textit{-}Max \ \mathit{DRVV}$ by restricting the solution of $\mathit{DRVV}$ to $k$ vertices using the Maximum Coverage Problem. This algorithm guarantees a $\frac{\rho \cdot r}{d}$ approximation ratio for the $k\textit{-}Max \ \mathit{DRVV}$ problem, where $\rho$ is the approximation ratio for $Max \ CP$, $r$ is the approximation ratio for $\mathit{DRVV}$, $d=\lceil m /k \rceil$, $k \leq m < n$ and $m=|S|$ with $S$ the solution provided by the algorithm $\mathit{DRVV}$. For the class of graphs with bounded maximum degree, the approximation ratio becomes $O(k/n)$. Additionally, we define the Dominating Set with Vulnerable Vertices ($\mathit{DSV}$) problem and show how it can be reduced to Red-Blue Set Cover. Using this reduction, we obtain a polynomial-time algorithm with a $2\sqrt{|V|\cdot(H(\Delta_{N})-\frac{1}{2}})$-approximation ratio, where $|V|$ is the order of the graph, $\Delta_N$ is the maximum degree among non-vulnerable vertices and $H$ is the harmonic function. Also, we consider the Vertex Cover problem with vulnerability constraints and define the Vertex Cover problem with vulnerable edges ($\mathit{VCVE}$). We propose an approximation algorithm that achieves a $4$-approximation guarantee and an LP rounding scheme further improving it to $2$-approximation guarantee.

This paper is organized as follows.
Section~\ref{sec:k-drvv} introduces the $k$-Vertex Maximum Domination Ratio with Vulnerable Vertices ($k\textit{-}Max~\mathit{DRVV}$) problem, establishes NP-hardness for the unbudgeted $DRVV$ problem, and develops approximation algorithms via the auxiliary $\mathit{DRVV}$ and $\mathit{DVV}$ formulations.
Section~\ref{sec:dsv} studies the Dominating Set with Vulnerable Vertices ($\mathit{DSV}$) problem and provides complexity and approximation results.
Section~\ref{sec:vcve} defines the Vertex Cover with Vulnerable Edges ($\mathit{VCVE}$) problem and presents a $2$-approximation algorithm.

\section{\texorpdfstring {$k$-Vertex Maximum Domination Ratio with Vulnerable Vertices ($k\textit{-}Max \ \mathit{DRVV}$)}{k-Vertex Maximum Domination Ratio with Vulnerable Vertices (k-Max DRVV)}}
\label{sec:k-drvv} 
Inspired by the problem defined in \cite{CHEN202184}, we propose a variant termed the $k$-Vertex Maximum Domination Ratio with Vulnerable Vertices ($k\textit{-}Max \ \mathit{DRVV}$).
In contrast to \cite{CHEN202184}, which uses the Independent Cascade Model to estimate expected influence, our work focuses on maximizing the domination ratio on a deterministic graph (see Fig.~\ref{fig:drvv}).

\begin{figure}
    	\centering
    	\includegraphics[scale=1.3]{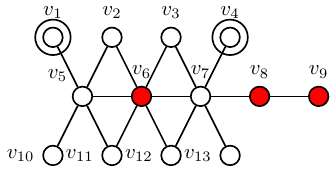}
    	\caption{A $k\textit{-}Max \ \mathit{DRVV}$ instance. Red vertices $v_6,v_8,v_9$ are vulnerable while the other ones are all non vulnerable. The circled vertices define a solution $S=\{v_1, v_4\}$ for $k=2$. For this instance $ASR(S, c)=\frac{4+c}{c}$}.
    	\label{fig:drvv} 
    \end{figure}
    
We define $ASR$ as follows: Given a graph $G=(V,E)$  where vertices are partitioned into non-vulnerable $V_N$ and vulnerable $V_L$, the  $ASR$ metric of a subset of vertices $S \subseteq V_N$ is $ASR(S,c) =\frac{\sigma_N(S) + c}{\sigma_L(S) + c}$, where $c$ is a positive constant, $\sigma_N(S)$ is the count of non-vulnerable vertices in the closed neighborhood of $S$, and  $\sigma_L(S)$ is the count of vulnerable vertices in the  neighborhood of $S$. The closed neighborhood of $S$ is the union of $S$ and the set of adjacent vertices to $S$. With these in mind, we proceed to the problem definition:

\begin{definition}[$k\textit{-}Max \ \mathit{DRVV}$]
Given a graph $G=(V, E)$  where vertices are partitioned into non-vulnerable $V_N$ and vulnerable $V_L$, an integer $k \leq |V_N|$ and $c \in \mathbb{R}^+$ a constant, find a subset $S\subseteq V_N$ such that $ASR(S,c)$ is maximized and $|S| \leq k$.
\end{definition}

\subsection{\texorpdfstring {Maximum Domination Ratio with Vulnerable Vertices ($\mathit{DRVV}$)}{Maximum Domination Ratio with Vulnerable Vertices (DRVV)}}
\label{sec:drvv} 
In this subsection, we introduce the Maximum Domination Ratio with Vulnerable Vertices Problem, as an unbudgeted version of $k\textit{-}Max \ \mathit{DRVV}$. First, we prove that $\mathit{DRVV}$ is NP-hard by reduction from the $3$-$\mathcal{SAT}_{equal}$ \cite{lamprou2021maximum}. Next, we define the Maximum Domination with a Bounded Number of Vulnerable Vertices problem, we prove its hardness and present an approximation algorithm for this problem. Finally, we use the result for $\mathit{DVV}$ to obtain an approximation result for $\mathit{DRVV}$.

\begin{definition}[$\mathit{DRVV}$]
	Given a graph $G=(V, E)$  where vertices are partitioned into non-vulnerable $V_N$ and vulnerable $V_L$, $c \in \mathbb{R}^+$ a constant, find a subset $S\subseteq V_N$ such that $ASR(S,c)$ is maximized.
\end{definition}

To prove that $\mathit{DRVV}$ is NP-hard, we define its decision version and reduce the NP-complete problem $3$-$\mathcal{SAT}_{equal}$ to decision $\mathit{DRVV}$\cite{lamprou2021maximum}.

\begin{definition}[Decision $\mathit{DRVV}$]
	Given a graph $G=(V, E)$ where vertices are partitioned into non-vulnerable $V_N$ and vulnerable $V_L$ and a rational number $W$ and $c \in \mathbb{R}^+$ a constant, decide whether exists a subset of vertices $S\subseteq V_N$ such that it holds  $ASR(S,c)\geq W$.
\end{definition}

\begin{definition}[$3$-$\mathcal{SAT}_{equal}$ \cite{lamprou2021maximum}]
	Given a CNF formula $\phi$ with $n$ variables and $n$ clauses, where each clause is a disjunction of exactly 3 literals, decide whether $\phi$ is satisfiable.
\end{definition}

 \begin{figure}
    	\centering
    	\includegraphics[scale=1]{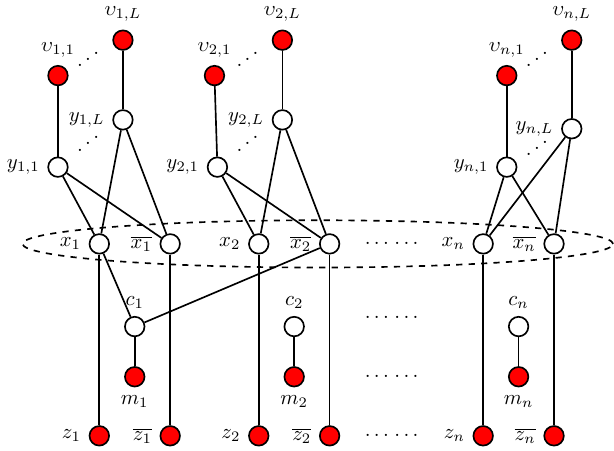}
    	\caption{The graph $G$ constructed for the reduction. The dashed ellipsis is a clique. All leaves (red vertices) are vulnerable.}
    	\label{fig:red} 
    \end{figure}
    
    \textbf{The Reduction.}
    Given a $3$-$\mathcal{SAT}_{equal}$ formula $\phi$, we create a graph $G$.
    Let $x_1,\overline{x_1}, x_2,\overline{x_2}, \ldots, x_n,\overline{x_n}$ stand for the variables of $\phi$ and $c_1, c_2, \ldots, c_n$ for the clauses of $\phi$ (see Fig.~\ref{fig:red}). We construct the graph $G$ in the following way:
    we place one vertex per literal $x_i, \overline{x_i}$ ($2n$ vertices in total), 
    one vertex per clause $c_i$ ($n$ vertices) and a set of $L$ vertices, where $L$ is $poly(n)$ (definition of $L$ follows later), for each variable (namely $y_{i j}$ for $j = 1, \ldots, L$) summing up to $L \cdot n $ vertices. We call the two vertices $x_i, \overline{x_i}$ a \emph{literal-pair} and each vertex $c_i$ a \emph{clause-vertex}. We add vulnerable vertices $z_i, \overline{z_i}$ connected to literals,  vulnerable vertices $m_i$ connected to clauses and each $y_{i j}$ vertex is connected to a vulnerable vertex $\upsilon_{i j}$.
    Then, we connect each literal vertex to \emph{all} the other literal vertices. Moreover, each literal-vertex is connected to all the corresponding clause-vertices where it appears in $\phi$. Finally, $x_i$ and $\overline{x_i}$ are connected to $y_{i j}$ for all $j$. The constant $c$ is part of the input and does not affect the proof. It is clear that the construction can be done in polynomial time.  

Before we start the proof, we provide some notations. The open neighborhood of $S$ is defined as $N(S) = \{ v \in V : \exists u \in S \text{ such that } (u,v) \in E \}$ and the closed neighborhood as $N[S]=N(S)\cup S$. Let 
    $C=\{c_1,...,c_n\}$ be the set of all clauses, 
    $X=\{x_1, \overline{x_1},...,x_n, \overline{x_n}\}$ the set of all variables, 
    $C(S) = \{v\in C: v\in N(S)\}$ the set of clauses that are adjacent to S, 
    $X(S) = \{v\in X: v\in N(S)\}$ the set of variables that are adjacent to S, 
    $C_S(x) = \{v\in C: v\in N(\{x\}) \ \wedge \ v\notin  C(S)\}$ the set of clauses that are adjacent to variable $x\in X$ and not adjacent to $S$, 
    $Y=\{y_{i j}\in V: \forall i,j\}$ the set of all $y_{i j}$ vertices as defined above and
    $Y_i = \{v\in V: v=y_{i j}, \ \forall j\}$ the subset of vertices $Y$ that are adjacent to variable $x_i\in X$.
    Let $\sigma_N(S)=|N[S]\cap V_N|$, $\sigma_L(S)=|N[S]\cap V_L|$ and 
    $ASR(S,c) = \frac{\sigma_N(S) + c}{\sigma_L(S) + c}, \ S\subseteq V \ and \ c \in \mathbb{R}^+$ a constant. We define $\bar{c}=\Big\lceil\max\{c,\frac{1}{c}\}\Big\rceil$, $L=\lceil(1+\bar{c}+2n\bar{c})n\rceil$ and $\boldsymbol{\mathcal{S}} = \{S\subseteq V: ASR(S,c)\geq ASR(S',c), \ \forall S'\subseteq V\}$ be the collection of all subsets of vertices $V$ that maximize the function  $ASR$. By definition, $|Y_i|=L$ for all $i\in \{1,\cdots,n\}$. 
    
    To prove the hardness result for $\mathit{DRVV}$, we first show that: if there exists a satisfiable assignment for $3$-$\mathcal{SAT}_{equal}$ then the value of $ASR(S,c)$ is greater than a specific value $W = \frac{Ln+3n+c}{n+c}$. 
            \begin{lemma}\label{lem:ASR>W}
    	If $\phi$ is satisfiable, there exists $S\subseteq V$ such that \resizebox{.26\textwidth}{!}{$ASR(S,c) \geq \frac{Ln+3n+c}{n+c}$}.
    \end{lemma}
 \begin{proof}
    Let $S=\{s_1,s_2,\cdots,s_n\}$ be a satisfiable assignment for $3$-$\mathcal{SAT}_{equal}$ problem, where $s_i=x_i \ or \ \overline{x_i}$, for every $  i=\{1,...,n\}$. By reduction, every clause $c_i\in C$ has at least one neighbor in $S$.
    Because every non-vulnerable vertex has only one vulnerable neighbor, $\sigma_L(S) = |S|$.  So,
    	\[
    	ASR(S,c)=\frac{\sigma_N(S) + c}{\sigma_L(S) + c}\geq\frac{|Y_i|\cdot n+|C(S)|+|X| + c}{|S|+ c} = \frac{Ln+3n+c}{n+c}
    	\] 
    \end{proof}
    
    We continue the proof by showing that, if there does not exists a satisfiable assignment for $3$-$\mathcal{SAT}_{equal}$ then the value of $ASR(S,c)$ is strictly lower than $W$. At first we show that a maximum size solution $S\in \boldsymbol{\mathcal{S}}$ contains only elements of $x_i \ or \  \overline{x_i}$, that is $S =\{s_1, s_2, ..., s_n\}$, where $s_i\in \{x_i,\overline{x_i}\}$ for all $i$.
    
     \begin{lemma}\label{lem:con}
   Let $S\in \boldsymbol{\mathcal{S}}$ be a subset of vertices that maximize the function  $ASR$. It holds $|S|=n$ and $S=\{s_1, s_2, \cdots, s_n\}$, where $s_i\in \{x_i,\overline{x_i}\}$ for all $i\in\{1,...,n\}$.
    \end{lemma}

\begin{proof} 
     It will be sequentially shown the next claims:
    \begin{enumerate}
     \item[] $\mathbf{Claim \ 1:}$ Every maximum size solution contains at least one vertex from set $X$.
      \item[] $\mathbf{Claim \ 2:}$ Every maximum size solution is a subset of $X$ vertices, that is $S\subseteq X$.
      \item[] $\mathbf{Claim \ 3:}$ Let $S\subseteq V$ with the property $S\cap X\neq\emptyset$ and $i \in \{1,2,\cdots,n\}$. If $ x\in \{x_i, \overline{x_i}\}$ then:
    	\[
    	\sigma_N(S\cup \{x\}) = \begin{cases}
    		|Y_i|+|C_{S}(x)|+\sigma_N(S) &  
    		, x\notin S \ and \ S\cap\{x_i, \overline{x_i}\}=\emptyset \\\\
    		\sigma_N(S) &  , x\in S \\\\
    		|C_{S}(x)|+\sigma_N(S) & , otherwise\\
    	\end{cases}
    	\]
      \item[] $\mathbf{Claim \ 4:}$ Every maximum size solution has exactly n vertices, that is $|S|=n$.
    \end{enumerate}
    \noindent$\mathbf{Claim \ 1 \ proof:}$

    $\bullet$ First, we show that $S$ is not empty.
    
    \noindent Let $S=\emptyset$. Then:
    	
    	\[
    	ASR(S,c) = ASR(\emptyset,c) = \frac{\sigma_N(S) + c}{\sigma_L(S) + c} =  \frac{0 + c}{0 + c} = 1
    	\]
    	We have a contradiction because:
    	\[
    	ASR(\{x_1\},c) = \frac{|Y_1|+  |C_{\emptyset}(x_1)| +|X|+ c}{1 + c} = \frac{|Y_1|+ |C_{\emptyset}(x_1)| + c}{1 + c} + \frac{|X|}{1 + c} \Rightarrow
     \]
     \[
     ASR(\{x_1\},c) > 1 + \frac{2n}{1 + c} \geq ASR(S,c)
    	\]
     
    $\bullet$ We show that $S$ has at least one vertex from set $X$.	
    
    \noindent Suppose $S\neq \emptyset$  contains no variable $x\in X$ and so $\emptyset\subset S\subseteq C\cup Y$. We have:
    	
    	\[ASR(S,c) = \frac{\sigma_N(S) + c}{\sigma_L(S) + c} =  \frac{|S| + |X(S)| + c}{|S| + c} = 1 + \frac{|X(S)|}{|S| + c} \leq 1+\frac{2n}{1 + c}
    	\]
    	We have a contradiction because:
    	\[
    	ASR(\{x_1\},c) > 1 + \frac{2n}{1 + c} \geq ASR(S,c)
    	\]
    \hfill$\square$
    
    \noindent $\mathbf{Claim \ 2 \ proof:}$
    
    $\bullet$ First Step: We show that $S\cap Y=\emptyset$.
    
    \noindent Let $(\exists y\in Y)  \ y\in S$. Without loss of generality, let $y=y_{i,j}$. We show that if we remove the node $y$ from the set $S$, the value of the function $ASR$ decreases.
    	By graph structure, $\sigma_L(S\setminus \{y_i\})=\sigma_L(S)-1$. Also, it has been shown that $(\exists x\in X) \ x\in S$ and that means $x_i$ and $\overline{x_i}$ are already neighbors of $S$. So, we have two cases:
    	\begin{itemize}
    		\item If $x_i\in S$ or $\overline{x_i}\in S$ then, by removing $y_i$ from the set $S$, the value of $\sigma_N$ does not change because $y_i$ is already adjacent with $x_i$ or $\overline{x_i}$ and so $\sigma_N(S\setminus \{y_i\})=\sigma_N(S)$. We have:
    		\[
    		ASR(S\setminus \{y_i\}, c) = \frac{\sigma_N(S) + c}{\sigma_L(S)-1 + c} > \frac{\sigma_N(S) + c}{\sigma_L(S) + c} = ASR(S, c)
    		\]
    		\item If $x_i\notin S$ and $\overline{x_i}\notin S$ then, by removing $y_i$ from the set $S$, the value of $\sigma_N$ decreases by 1 because $y_i$ is not adjacent with $S$ and so $\sigma_N(S\setminus \{y_i\})=\sigma_N(S)-1$. We have:
    		\[
    		ASR(S\setminus \{y_i\}, c) = \frac{\sigma_N(S)-1 + c}{\sigma_L(S)-1 + c} > \frac{\sigma_N(S) + c}{\sigma_L(S) + c} = ASR(S, c)
    		\]
    	\end{itemize}
    	The above inequality stands because, by Claim 1 and graph structure, we have $ASR(S, c)>1$, $\sigma_N(S)-1+c > 0$ and $\sigma_L(S)-1+c > 0$. So, we have a contradiction.\\
     
        $\bullet$ Second Step: We show that $S\cap C = \emptyset$.
        
        \noindent Let $\exists v\in C$ such that $  \ v\in S$ and let $v=c_i$. We show that if we remove the node $v$ from the set $S$, the value of the function $ASR$ decreases. By graph structure, $\sigma_L(S\setminus \{c_i\})=\sigma_L(S)-1$. So, we have two cases:
    	\begin{itemize}
    		\item If $N(c_i)\cap S \neq \emptyset$ then, in the same way with first step, $\sigma_N(S\setminus \{c_i\})=\sigma_N(S)$ and so
    		\[
    		ASR(S\setminus \{c_i\}) = \frac{\sigma_N(S) + c}{\sigma_L(S)-1 + c} > \frac{\sigma_N(S) + c}{\sigma_L(S) + c} = ASR(S)
    		\]
    		\item If $N(c_i)\cap S = \emptyset$ then, in the same way with first step, $\sigma_N(S\setminus \{c_i\})=\sigma_N(S)-1$ and since $ASR(S,c)>1$
    		\[
    		ASR(S\setminus \{c_i\}) = \frac{\sigma_N(S)-1 + c}{\sigma_L(S)-1 + c} > \frac{\sigma_N(S) + c}{\sigma_L(S) + c} = ASR(S,c)
    		\]
    	\end{itemize}
    	So, we have a contradiction. 
    	\hfill $\square$\\
     
    \noindent$\mathbf{Claim \ 3 \ proof:}$      
        \begin{itemize}
                
                \item[$\bullet$] Let $x\in S$. Then, $S\cup \{x\} = S$ and so we have $\sigma_N(S\cup \{x\})=\sigma_N(S)$
                
                \item[$\bullet$] Let $x\notin S$ and $S\cap\{x_i, \overline{x_i}\}=\emptyset$. Then, we have:
                \begin{itemize}
                  \item[$(a)$] $X \subseteq N(S)$,  because, by Claim 1  there exists $ x \in X $ such that $x \in S$ and by definition $X \subseteq N(x)$.
                 \item[$(b)$]
                 $y_{ij}\notin N(S)$ for every $j$, because $S\cap\{x_i, \overline{x_i}\}=\emptyset$.
                 \item[$(c)$]
                 $C_{S}(x)$ is the set of all $c\in C$ with the property $c\notin N(S)$ and $c\in N(S\cup{\{x\}})$.                
                  \end{itemize}
                  So, comparing the neighbors of the sets $S\cup{\{x\}}$ and $S$, the set $S\cup{\{x\}}$ has as neighbors the neighbors of $S$ and additionally the nodes of the sets $Y_i$ and $C_{S}(x)$, that is $\sigma_N(S\cup \{x\}) = |Y_i|+|C_{S}(x)|+\sigma_N(S)$.
                
                \item[$\bullet$] Let $x\notin S$ and let $x = x_i$ and $\overline{x_i}\in S$. Then, we have the same condition as the previous case with the exception that for every $j$, $y_{ij}\in N(S)$ because $\overline{x_i}\in S$. So, $\sigma_N(S\cup \{x\}) = |C_{S}(x)|+\sigma_N(S)$.
        \end{itemize}
    \hfill $\square$\\
    
   \noindent$\mathbf{Claim \ 4 \ proof:}$  
   
   $\bullet$ First, we show that $S$ contains at most one of every pair of vertices $\{x_i, \overline{x_i}\}$.
    
    \noindent Let $x_i, \overline{x_i} \in S$ and $S' = S\setminus \{x_i, \overline{x_i}\}$. Without loss generality, if we show that $ASR(S)<ASR(S\setminus \{\overline{x_i}\})$, then the result stand by contradiction. By using claim 3, we have:
    	\[
    	ASR(S\setminus \{\overline{x_i}\},c)= ASR(S'\cup \{x_i\},c)=\frac{\sigma_N[S'\cup \{x_i\}] + c}{\sigma_L(S'\cup \{x_i\}) + c}  \iff
    	\]
        \[
    	ASR(S\setminus \{\overline{x_i}\},c)= \frac{|Y_i|+|C_{S'}(x_i)|+\sigma_N(S') + c}{1+|S'| + c}
    	\]
    	Also, because $ASR(S,c)= ASR((S' \cup \{x_i\})\cup \ \{\overline{x_i}\},c)$, we have:
    	\[
    	ASR(S,c)= \frac{|C_{S'\cup \{x_i\}}(\overline{x_i})|+\sigma_N(S'\cup \{x_i\}) + c}{1+|S'\cup \{x_i\}| + c}
    	  \iff \]
    	\[
    	ASR(S,c)= \frac{|Y_i|+|C_{S'}(x_i)|+|C_{S'\cup \{x_i\}}(\overline{x_i})|+\sigma_N(S') + c}{2+|S'| + c}
    	\]
    	We have:
    	\[
    	ASR(S,c)<ASR(S\setminus \{\overline{x_i}\},c) \iff
    	\]
    	\[
    	\frac{|Y_i|+|C_{S'}(x_i)|+|C_{S'\cup \{x_i\}}(\overline{x_i})|+\sigma_N(S') + c}{2+|S'| + c}< \frac{|Y_i|+|C_{S'}(x_i)|+\sigma_N(S') + c}{1+|S'| + c}
    	\]
        \[
        \iff
        \]
    	\[
    	(1+|S'| + c)|C_{S'\cup \{x_i\}}(\overline{x_i})| < |Y_i|+|C_{S'}(x_i)|+\sigma_N(S')+c
    	\]
    	The inequality stands because $S'\subseteq X \then \ |S'|\leq |X| = 2n$,\\
        $C_{S'\cup \{x_i\}}(\overline{x_i})\subseteq C \then |C_{S'\cup \{x_i\}}(\overline{x_i})|\leq |C| = n$ and so:
    	\[
    	|Y_i|+|C_{S'}(x_i)|+\sigma_N(S')+c > |Y_i| = (1+\bar{c}+2n\bar{c})n \geq (1+2n+c)n \iff
    	\]
        \[
    	|Y_i|+|C_{S'}(x_i)|+\sigma_N(S')+c >  (1+|S'| + c)|C_{S'\cup \{x_i\}}(\overline{x_i})|
    	\]

    $\bullet$ In the last step, we show that $S$ has exactly $n$ elements.\\   \\ 
    \noindent $-$ Let $S$ has more than $n$ element, that is $|S|>n$. By claim 2, $S\subseteq X$. So, $\exists i=\{1,...,n\}$ such that $x_i, \overline{x_i}\in S$. We have a contradiction.\\  \\   
    \noindent $-$ Let $|S|=k < n$. Then, $(\exists x_i\in V) \ x_i\notin S$. If we show that $ASR(S,c)<ASR(S\cup \{x_i\},c)$, then the results stand by contradiction. We have:
    	\[
    	ASR(S,c)<ASR(S\cup \{x_i\},c) \iff
    	\]
    	\[
    	\frac{\sigma_N(S) + c}{|S| + c}< \frac{|Y_i|+|C_{S}(x_i)|+\sigma_N(S) + c}{1+|S| + c}\iff
    	\]
    	\[
    	\sigma_N(S)+c<(|Y_i|+|C_{S}(x_i)|)(|S|+c)
    	\]
    	The inequality stands because:
    	\[
    	(|Y_i|+|C_{S}(x_i)|)(|S|+c) \geq |Y_i|(k+c) = |Y_i|k+|Y_i|c = |Y_i|k+(1+2n\cdot \bar{c}+\bar{c})\cdot n\cdot c \geq\]
     \[
     |Y_i|k +\bar{c}\cdot c \cdot n + +\bar{c}\cdot c \cdot 2n^2+ n\cdot c > |Y_i|k + n +2n + c \iff
    	\]
    	\[
    	(|Y_i|+|C_{S}(x_i)|+ c)(|S|+c) > |Y_i|k + n + 2n + c \geq |Y_i|k + |C(S)|+2n + c = \sigma_N(S) + c
    	\]
        
\end{proof} 

    \begin{lemma}\label{lem:ASR<W}
    	If there exists no satisfiable assignment for $\phi$, then 
        $\forall S \subseteq V$ it holds $ASR(S,c) < \frac{Ln+3n+c}{n+c}$.
    \end{lemma}

     \begin{proof}
    	Let $S\in \boldsymbol{\mathcal{S}}$ be a subset of vertices that maximize the function  $ASR$. By Lemma~\ref{lem:con}, it holds $S =\{s_1, s_2, ..., s_n\}, \ s_i = x_i \ or \  \overline{x_i}, \ \forall i=\{1,...,n\}$, there exists $c_i\in C$ such that $c_i\notin N(S)$.
    	
    	Because there exists no satisfiable assignment for $3$-$\mathcal{SAT}_{equal}$ there exists a clause $c_i\in C$ that is not satisfiable. By reduction, for every set $S\in \boldsymbol{\mathcal{S}}$, we have:
    	
    	\[
    	ASR(S,c) = \frac{|Y_i|\cdot n+|C(S)|+2n+c}{n+c}\leq \frac{Ln+n-1+2n+c}{n+c}<\frac{Ln+3n+c}{n+c}
    	\]  
    \end{proof}
    
\begin{theorem}\label{th:drvv_np_complete}
Decision $\mathit{DRVV}$ is NP-Complete.
\end{theorem} 
\begin{proof}
    Decision $\mathit{DRVV}$ is in NP: given a candidate subset $S \subseteq V$, the value $ASR(S,c)$ can be computed in polynomial time, so membership can be verified in polynomial time.

    For NP-hardness, we reduce from $3$-$\mathcal{SAT}_{equal}$. Given an instance $\phi$, construct the graph $G$ and threshold $W = \frac{Ln+3n+c}{n+c}$ as described above. By Lemma~\ref{lem:ASR>W}, if $\phi$ is satisfiable then there exists $S \subseteq V$ with $ASR(S,c) \geq W$, so the decision instance answers \emph{yes}. By Lemma~\ref{lem:con} and Lemma~\ref{lem:ASR<W}, if $\phi$ is not satisfiable then every $S \subseteq V$ satisfies $ASR(S,c) < W$, so the decision instance answers \emph{no}. The construction runs in polynomial time, completing the reduction.
    \end{proof}

Theorem~\ref{th:drvv_np_complete} implies the hardness of $k\textit{-}Max \ \mathit{DRVV}$. Indeed, the existence of a polynomial-time algorithm for it would have led to a polynomial algorithm for $\mathit{DRVV}$.

\begin{theorem}\label{th:k-drvv_np_complete}
$k\textit{-}Max \ \mathit{DRVV}$ is NP-Hard.
\end{theorem} 
      
\subsection{\texorpdfstring {Maximum Domination with Bounded number of Vulnerable Vertices ($\mathit{DVV}$)}{Maximum Domination with Bounded number of Vulnerable Vertices (DVV)}}\label{sec:DVV}

Following, we define the Maximum Domination with Bounded number of Vulnerable Vertices. The goal is to maximize the count of dominated non-vulnerable vertices while the vulnerable vertices are kept up to a bound. 

\begin{definition}[$\mathit{DVV}$]
	Given a graph $G=(V, E)$ and an integer $B$, where vertices are partitioned into non-vulnerable $V_N$ and vulnerable $V_L$, find a subset of vertices $S\subseteq V_N$ such that $S$ dominates at most $B$ vulnerable vertices $(\sigma_L(S)\leq B$, $0\leq B \leq |V_L|$ integer$)$ and maximizes the number of non-vulnerable vertices it dominates. 
\end{definition}

$\mathit{DVV}$ is NP-hard, and to prove this, we use a reduction from $k\textit{-}Max \ VD$, which was proven to be NP-hard \cite{miyano2011maximum}.

\begin{definition}[$k\textit{-}Max \ VD$ \cite{miyano2011maximum}]
	Given a graph $G (V, E)$ and an integer $k<|V|$ find a subset $S\subseteq V$ where $|S| \le k$ such that the count of dominated vertices is maximized. 
\end{definition}

\begin{theorem}\label{theor:DVV_hard}
    The $\mathit{DVV}$ is as hard as $k\textit{-}Max \ VD$.
\end{theorem}
\begin{proof}
Given an instance of $k\textit{-}Max \ VD$ with a graph $G=(V, E)$, for every vertex $v_i$, add a vertex $u_i$ and an edge $[v_i ,u_i]$. Let $V'$ be the set of new vertices and $F$ the set of new edges. For the augmented graph $G'=(V\cup V', E\cup F)$, $V$ is the set of non-vulnerable vertices and $V'$ is the set of vulnerable vertices. Let $S'$ be an optimal solution  of $\mathit{DVV}$ on $G'$. Since every vertex in $G'$ has by construction, exactly one vulnerable neighbor, setting $B=k$ ensures $S'$ has at most $k$ vertices and it dominates maximum count of non-vulnerable vertices. So, solution $S'$ for $\mathit{DVV}$ on $G'$ is exactly an optimal solution $S$ for $k\textit{-}Max \ VD$ on $G$ and vice versa.
\end{proof}

To address the $\mathit{DVV}$ problem, we formulate it as a Set Union Knapsack Problem ($SUKP$) \cite{ARULSELVAN2014214}. This problem admits a polynomial approximation algorithm with a constant approximation ratio of $ (1-e^{-\frac{1}{d}})$, when the number of items in which an element is present is bounded by a constant $d$.

\begin{definition}[$SUKP$]
The problem consists of a set of elements $U = \{1,...,n\}$ and a set of items $X=\{1,...,m\}$. Each item $i\in X$ corresponds to a subset of elements $X_i \subseteq U$ with a nonnegative profit given by $p : X \rightarrow\mathbb{R}^+$. Each element has a nonnegative weight given by $w : U \rightarrow\mathbb{R}^+$. For a subset $A \subseteq X$, the weighted union of set
$A$ is $W (A) = \sum_{e \in (\cup_{i \in A}X_{i})}w(e)$ and $P(A) = \sum_{i  \in A}p(i)$. The goal is to find a subset of items $S \subseteq X$ such that $P(S)$ is maximized and $W(S) \leq B$, where $B$ is a given budget. Let $(X,p,w,B)$ denote a $SUKP$ instance.
\end{definition}

\textbf{The Algorithm.} We transform each instance of $\mathit{DVV}$ into an instance of $SUKP$ (see Algorithm~\ref{alg:dvv}). At first, all vertices that have no vulnerable neighbors are added to the solution. The remaining vertices are used to construct an instance of the $SUKP$. Each vulnerable vertex corresponds to an element, which is assigned a weight of 1. Each non-vulnerable vertex is represented as an item, including the vulnerable neighbors it dominates. The profit of each item is set equal to the number of non-vulnerable vertices it dominates. Based on this mapping, a $SUKP$ instance $(X, p, w, B)$ is created and solved using a $SUKP$ algorithm provided by \cite{ARULSELVAN2014214}. The solution obtained from $SUKP$ is then translated back into a feasible solution for the $\mathit{DVV}$ problem. The complexity of this algorithm is dominated by the complexity of the $SUKP$ algorithm, which is $O(|V|^3)$. We prove that the  approximation guarantee for $\mathit{DVV}$ is  $\frac{1}{\Delta+1} (1-e^{-\frac{1}{\Delta}})$.

 \begin{algorithm}[H]
        \label{alg:dvv}
	\SetKwInOut{Input}{Input}
	\SetKwInOut{Output}{Output}
	\DontPrintSemicolon
	\Input{A graph $G=(V, E)$ and an integer $B \leq |V_L|$. The vertices $V$ are partitioned into non-vulnerable $V_N$ and vulnerable $V_L$}
	\Output{A subset of vertices $S\subseteq V_N$ such that $S$ is a feasible solution}
	
	$S_0 \leftarrow \emptyset$\; 
	\ForEach{$u \in V_L$}{
		$w(u) \leftarrow 1$\;
	}
	\ForEach{$\upsilon \in V_N$}{
		$X_\upsilon \leftarrow N(\upsilon)\cap V_L$\;
            \If{$X_\upsilon = \emptyset$}{
            
            $S_0 \leftarrow S_0 \cup \{\upsilon\}$\;
           
            }
	}
        \ForEach{$\upsilon \in V_N $}{
		$p(\upsilon) \leftarrow \sigma_N(\upsilon)$\;
	}
	$X \leftarrow V_N \setminus S_0$\;
        $S_{SUKP} \leftarrow SUKP(X,p,w,B)$\;
        
	$S \leftarrow S_0 \cup S_{SUKP}$\;
	\textbf{return} $S$
	
	\caption{$\mathit{DVV}(G,B)$}
\end{algorithm}

 \textbf{The Analysis.} We proceed with some notation. Let $S^*$ be an optimal solution of $\mathit{DVV}$, $S_{SUKP}^*$ an optimal solution of the $SUKP$ problem, $S_{SUKP}$ be a set of vertices that $SUKP$ algorithm returns, $S_0$ be the set of vertices with no vulnerable neighbors and $S$ the set of vertices the algorithm returns as the solution of $\mathit{DVV}$. Also $P(S)=\sum_{v \in S}p(v)$ be the sum of profits of a set of vertices $S$, $\Delta_{N}=\max_{\upsilon \in V_N}|N(\upsilon)\cap V_N| \leq \Delta$, $\Delta_{L} = \max_{\upsilon \in V_L} |N(\upsilon)\cap V_N| \leq \Delta$ and $N_A[u]$ is defined to be the closed neighborhood of vertex $u$ in subset $A$. Recall that $\sigma_N(S)$ is the count of non-vulnerable vertices in the closed neighborhood of $S$ (including $S$ itself). For the analysis all non-vulnerable vertices are counted in the profit function $p(\upsilon) = \sigma_N(\upsilon)$, for every $ \upsilon \in V_N$. 

\begin{lemma}\label{lemma:dvv_inequality}
For any subset $A \subseteq V $ it holds:
	$\sigma_N(A)\leq P(A) \leq (\Delta_{N}+1)\cdot \sigma_N(A).$
\end{lemma}
\begin{proof}
	By Algorithm~\ref{alg:dvv}, every profit of a vertex is equal to the number of non-vulnerable neighbors. So, every non-vulnerable vertex adds one point to each neighboring vertex profit, including itself. So:
    
	$$P(A)=\sum_{u \in A}p(u)
	=\sum_{u \in N[A] \cap V_N}\big|N_A[u]\big|$$
    
	By definition of $N_A[u]$, for every $u \in N[A] \cap V_N$, $1 \leq \big|N_A[u]\big| \leq \Delta_{N}+1$. So, 
	\[
	\sum_{u \in N[A] \cap V_N}1\leq
	\sum_{u \in N[A] \cap V_N}\big|N_A[u]\big|\leq
	\sum_{u \in N[A] \cap V_N}(\Delta_{N}+1)\Rightarrow
    \]
    \[
	\sigma_N(A)\leq
	P(A) \leq 
	(\Delta_{N}+1)\cdot\sigma_N(A)
	\]
\end{proof}

\begin{theorem}\label{thm:dvvProof}
Algorithm \ref{alg:dvv} is a
$	\frac{1}{\Delta+1}\cdot (1-e^{-\frac{1}{\Delta}})$ - approximation for $\mathit{DVV}$ problem.
\end{theorem}

\begin{proof}
    There exists an optimal solution $S^*$ such that $S_0 \subseteq S^*$.  
	Let $(X,p,w,B)$ an instance of the $SUKP$ problem that Algorithm~\ref{alg:dvv} defines. By \cite{ARULSELVAN2014214}, $S_{SUKP}$ is a feasible solution for $SUKP$ instance and $(1-e^{-\frac{1}{\Delta_L}})\cdot P(S_{SUKP}^*) \leq P(S_{SUKP})$. By definition, $S^* \setminus S_0$ is also a feasible solution for $SUKP$ instance and so $P(S^* \setminus S_0) \leq  P(S_{SUKP}^*)$. We have:
 \[
 (1-e^{-\frac{1}{\Delta_L}})\cdot P(S^* \setminus S_0) \leq (1-e^{-\frac{1}{\Delta_L}})\cdot P(S_{SUKP}^*) 
\leq P(S_{SUKP}) \Rightarrow
 \]
 \[
(1-e^{-\frac{1}{\Delta_L}})\cdot (P(S^* \setminus S_0) +P(S_0)) \leq P(S_{SUKP})+P(S_0) \Rightarrow
 \]
\[
(1-e^{-\frac{1}{\Delta_L}})\cdot P(S^*) \leq P(S).
 \]
 
\noindent By Lemma~\ref{lemma:dvv_inequality}, $\sigma_N(S^*) \leq P(S^*)$ and $P(S)\leq(\Delta_{N}+1)\cdot \sigma_N(S)$. So: 
	\[
	(1-e^{-\frac{1}{\Delta_L}})\cdot \sigma_N(S^*) \leq (1-e^{-\frac{1}{\Delta_L}})\cdot P(S^*) \leq P(S) \leq (\Delta_{N}+1)\cdot \sigma_N(S)\Rightarrow
        \]
	\[
	\frac{1}{\Delta+1}\cdot (1-e^{-\frac{1}{\Delta}})\cdot \sigma_N(S^*)\leq \frac{1}{\Delta_{N}+1}\cdot (1-e^{-\frac{1}{\Delta_L}})\cdot \sigma_N(S^*)
	\leq \sigma_N(S)
	\]
\end{proof}

\subsection{\texorpdfstring {Algorithm and Approximation for $\mathit{DRVV}$}{Algorithm and Approximation for DRVV}}
\label{sec:algorithm} 

In order to get an approximation algorithm for $\mathit{DRVV}$ (Algorithm~\ref{alg:def8Algorithm}), we are using  Algorithm~\ref{alg:dvv}. We fix the number of vulnerable vertices $(B)$ and we create an instance of $\mathit{DVV}$. Then we call the $\mathit{DVV}$ algorithm, for $B$ from 0 to $|V_L|$. The value giving the best ratio is chosen, with the corresponding solution.

\begin{algorithm}[H]
	\SetKwInOut{Input}{Input}
	\SetKwInOut{Output}{Output}
	\DontPrintSemicolon
	\label{alg:def8Algorithm}
	\Input{A graph $G (V, E)$  where vertices are partitioned into vulnerable $V_L$ and non-vulnerable $V_N$ and a positive constant $c$}
	\Output{A solution $S \subseteq V_N$}	
	\ForEach {$B\in \{0,1,\cdots,|V_L|\}$} {
		$S_B \gets DVV(G,B)$
	}
    $S \gets \argmax\limits_{S_B \in \{S_0, S_1, \cdots, S_{|V_L|}\}}(ASR(S_B,c))$\;
	\textbf{return} $S$
	
	\caption{$\mathit{DRVV}$ Algorithm}
\end{algorithm}

\begin{theorem}\label{thm:drvvProof}
Algorithm~\ref{alg:def8Algorithm} is a $\frac{1}{\Delta+1}\cdot (1-e^{-\frac{1}{\Delta}})$ -approximation for $\mathit{DRVV}$ problem.
\end{theorem}
\begin{proof}
 Let $S^*_{DRVV}$ be the optimal solution of $\mathit{DRVV}$ and $B^*=\sigma_L(S^*_{DRVV})$. Let $S_{DVV}$ be the solution returned by Algorithm~\ref{alg:dvv} and $S^*_{DVV}$ be the optimal solution of $\mathit{DVV}$ for the same instance $G,B^*$. If $S_{DRVV}$ is the solution of Algorithm~\ref{alg:def8Algorithm} we have the following.
As described in Algorithm~\ref{alg:def8Algorithm}, the $\mathit{DVV}$ algorithm is called for every value of $B$, including the parameter $B^*$ corresponding to the optimal solution $S^*_{DRVV}$. So, it stands $\sigma_N(S^*_{DRVV})=\sigma_N(S^*_{DVV})$ by the definition of $\mathit{DVV}$. The Algorithm~\ref{alg:dvv} returns a set $S_{DVV}$ such that $r \cdot \sigma_N(S^*_{DVV}) \leq \sigma_N(S_{DVV})$, (where $r\leq 1$ denotes the approximation ratio guaranteed by Algorithm~\ref{alg:dvv}), and the cost constraint is satisfied, i.e., $\sigma_L(S_{DVV})\leq B^* = \sigma_L(S^*_{DRVV})$. Because of these we have:
	\[ASR(S_{DRVV},c) \geq ASR(S_{DVV},c)=\frac{\sigma_N(S_{DVV}) + c}{\sigma_L(S_{DVV}) + c} \geq \frac{r \cdot \sigma_N(S^*_{DVV}) + c}{\sigma_L(S_{DVV}) + c} = r \cdot ASR(S^*_{DRVV},c)
 \]
 So
 \[
 ASR(S_{DRVV},c) \geq r \cdot ASR(S^*_{DRVV},c)\] 
\end{proof}

We now return to the $k\textit{-}Max \ \mathit{DRVV}$ problem. Before proceeding, let's recall the definition of the $Max \ CP$ problem.

\begin{definition}[$Max \ CP$]
Given a universe $U = \{e_1, e_2, \ldots, e_n\}$ of $n$ elements, a collection of subsets $\mathcal{X} = \{X_1, X_2, \ldots, X_m\}$ where each $X_i \subseteq U$, and an integer $k$, the goal is to select a subcollection $\mathcal{C} \subseteq \mathcal{X}$ such that: The collection has at most $k$ subsets $(|\mathcal{C}| \leq k)$ and the union of the sets in $\mathcal{C}$ maximizes the number of covered elements in $U$
$\Big(\mathcal{C} = \argmax_{\mathcal{A} \subseteq \mathcal{X}, \ |\mathcal{A}| \leq k} \left| \bigcup_{X_j \in \mathcal{A}} X_j \right| \Big)$. 
\end{definition}

\subsection{\texorpdfstring {Algorithm and Approximation for \boldmath $k\textit{-}\mathrm{Max\ \mathit{DRVV}}$}{Algorithm and Approximation for k-Max DRVV}}
\label{sec:k-algorithm}

Given a solution $S$ to the $\mathit{DRVV}$ problem, Algorithm~\ref{alg:def1Algorithm} constructs a $Max \ CP$ instance such that, for every vertex $u\in S$, a subset of vertices $X_u= N[u]\setminus V_L$ is created, which includes all neighboring vertices of $u$ that are not vulnerable, including $u$ itself. The universe $U=N[S]\setminus V_L$ includes the close neighborhood of $S$ excluding the vulnerable vertices. In this formulation, the vertices to be selected are restricted to $S$, whereas the dominated vertices correspond to the closed neighborhood of $S$. The algorithm then solves this $Max \ CP$ instance using the standard greedy algorithm \cite{nemhauser1978analysis} and returns the resulting subset. The $Max \ CP$ problem has an approximation ratio $(1-1/e)$  \cite{Hochbaum,nemhauser1978analysis} and this bound cannot be improved unless $P=NP$ \cite{feige1998threshold}.

Let $S$ be the output of the $\mathit{DRVV}$ algorithm and $S^*$ be the optimal solution of $\mathit{DRVV}$, $S_k$ be the output of the $k\textit{-}Max \ \mathit{DRVV}$ algorithm, $S^*_k$ be the optimal solution of $k\textit{-}Max \ \mathit{DRVV}$ and $S'_k$ be the optimal solution for $Max \ CP$. We prove bicriteria Lemma~\ref{lemma:l1} holds and then we convert this result to a true approximation.

\begin{lemma}
\label{lemma:l1}
 Algorithm~\ref{alg:def8Algorithm}  is a $(d,r)$ bicriteria algorithm for $k\textit{-}Max \ \mathit{DRVV}$ problem, where $r$ is the approximation ratio for the unbudgeted $\mathit{DRVV}$, $d=\lceil m /k \rceil$, $k \leq m < n$ and $m=|S|$, $S$ the solution provided by algorithm~\ref{alg:def8Algorithm}.

\end{lemma}

\begin{proof}
 Notice that $S$ does not include vulnerable vertices. In the case where $m \leq k$ we have a feasible solution for $k\textit{-}Max \ \mathit{DRVV}$ problem with the approximation ratio that of the $\mathit{DRVV}$ Algorithm~\ref{alg:def8Algorithm}. In the case where $m=n$, it means that we have an instance without vulnerable vertices and the approximation is that of $Max \ CP$. So, we focus
in the case where $k \leq m < n$. By definition of $S^*$, it holds $ASR(S^*_k,c)\leq ASR(S^*,c)$ and $r \cdot ASR(S^*,c)\leq ASR(S,c)$. We have:
\[
r \cdot ASR(S^*_k,c) \leq r \cdot ASR(S^*,c)\leq ASR(S,c)\Rightarrow
 r \cdot ASR(S^*_k,c)\leq ASR(S,c)
\]

 Since $\mathit{DRVV}$ algorithm gives a subset S with ASR within $r$ of the optimum and this subset has a size $m < n$, the lemma stands.
\end{proof}

\begin{algorithm}[H]
	\SetKwInOut{Input}{Input}
	\SetKwInOut{Output}{Output}
	\DontPrintSemicolon
	
	\Input{A graph $G=(V, E)$ where vertices are partitioned into non-vulnerable $V_N$ and vulnerable $V_L$, an integer $k \leq |V|$ and a positive constant $c$}
	\Output{A subset $S_k$, such that $|S_k| \leq k$}
		
	$S \gets DRVV(G)$\;
	\uIf{$|S| \leq k$}
    {$S_k \gets S$}
    \Else{
    $U\gets N[S] \setminus V_L$\;
    \ForEach {$u\in S$}{
    $X_u = N[u]\setminus V_L$\;
    }
    $\mathcal{X}=\bigcup_{u\in S}\{X_u\}$\;
	$C=Max \ CP(U,\mathcal{X},k)$\;
    $S_k \gets \bigcup_{X_u \in C}\{u\}$\;
	}
	\textbf{return} $S_k$
	
	\caption{$k\textit{-}Max \ \mathit{DRVV}$}
	\label{alg:def1Algorithm}
\end{algorithm}
 \vspace{1em}

\textbf{Converting the bicriteria approximation to a true approximation.}

We need to get a subset of $S$, of size at most $k$ for a standard approximation algorithm. 

\begin{theorem}
 Algorithm~\ref{alg:def1Algorithm} is a $\frac{\rho \cdot r}{d}$ approximation for $k\textit{-}Max \ \mathit{DRVV}$ problem, where $\rho$ is the approximation ratio for $Max \ CP$, $r$ is the approximation ratio for $\mathit{DRVV}$, $d=\lceil m /k \rceil$, $k \leq m < n$ and $m=|S|$ with $S$ the solution provided by algorithm~\ref{alg:def8Algorithm}.

\end{theorem}
\begin{proof}
 
Let $m=|S| < n$ and $d=\lceil \frac{m}{k} \rceil$. By definition of $\sigma_N$, for any partition  $\{S_1,S_2,\cdots,S_d\}$ of $S$, such that $|S_i|\leq k$ for all $i=\{1,\cdots,d\}$, and $|S_1|+|S_2|+\cdots+|S_{d}|=m$, it holds $\sigma_N(S)\leq \sigma_N(S_1)+\cdots \sigma_N(S_d)$. Clearly, there is a $S_{i^*}$,$1 \leq i^* \leq d$ such that $\sigma_N(S_1)+\cdots+\sigma_N(S_d)\leq d\cdot\sigma_N(S_{i^*})$. We have:

\[
\sigma_N(S) \leq d \cdot \sigma_N(S_{i^*})\Rightarrow
\frac{\sigma_N(S)+c}{\sigma_L(S)+c} \leq \frac{d \cdot\sigma_N(S_{i^*})+c}{\sigma_L(S)+c}\Rightarrow
\]
\begin{equation}
\label{eq:2}
ASR(S,c) = \frac{\sigma_N(S)+c}{\sigma_L(S)+c} \leq \frac{d \cdot\sigma_N(S_{i^*})+c}{\sigma_L(S)+c}
\end{equation}

By Lemma~\ref{lemma:l1} it holds $ r \cdot ASR(S^*_k,c)\leq ASR(S,c)$. By Equation~\ref{eq:2}, we have:
\begin{equation}
\label{eq:3}
r\cdot ASR(S^*_k,c)\leq \frac{d \cdot\sigma_N(S_{i^*})+c}{\sigma_L(S)+c}\leq d \cdot\frac{\sigma_N(S_{i^*})+c}{\sigma_L(S)+c}
\end{equation}

By definition of $S'_k$, it holds $\sigma_N(S_{i^*})\leq\sigma_N(S'_k)$. Also, by \cite{Hochbaum,nemhauser1978analysis} it holds $\rho\cdot \sigma_N(S'_k)\leq\sigma_N(S_k)$. By Equation~\ref{eq:3}, we have:

\begin{equation}
\label{eq:4}
\frac{\rho\cdot r}{d}\cdot ASR(S^*_k,c)\leq \frac{\rho\cdot\sigma_N(S_{i^*})+\rho\cdot c}{\sigma_L(S)+c}\leq \frac{\rho\cdot\sigma_N(S'_k)+c}{\sigma_L(S)+c} \leq
\frac{\sigma_N(S_k)+c}{\sigma_L(S)+c}
\end{equation}

Furthermore, by definition of $\sigma_L$ it holds $\sigma_L(S_k) \leq \sigma_L(S)$ since $S_k \subseteq S$. So by Equation~\ref{eq:4}, we have:

\[
\frac{\rho \cdot r}{d} ASR(S^*_k,c) \leq \frac{\sigma_N(S_k)+c}{\sigma_L(S)+c} \leq
\frac{\sigma_N(S_k)+c}{\sigma_L(S_k)+c}=ASR(S_k,c) 
\]
\end{proof}

Finally, under the current best known results for $Max \ CP$ and $\mathit{DRVV}$ we obtain for $k\textit{-}Max \ \mathit{DRVV}$ problem the approximation ratio $\frac{1}{\Delta+1}\cdot (1-e^{-\frac{1}{\Delta}}) \cdot (1-\frac{1}{e})\cdot \frac{k}{n}$. We have the following corollary.

\begin{corollary} 
For the class of graphs with bounded maximum degree $\Delta$, Algorithm~\ref{alg:def1Algorithm} achieves
an approximation ratio of $O(k/n)$.
\end{corollary}

\section{\texorpdfstring {Dominating set with vulnerable vertices ($DSV$)}{Dominating set with vulnerable vertices (DSV)}}
\label{sec:dsv} 

In contrast to the previous problems where the focus was on maximizing a specific ratio, in this section we explore a different variant. The objective is no longer to compute a ratio, but rather to find a dominating set for the non-vulnerable vertices which dominates the minimum number of vulnerable vertices.
This problem formulation is closely related to the Red-Blue Set Cover ($RBSC$) problem, where the goal is to cover all the "blue" elements (representing the non-vulnerable vertices) while keeping the coverage of "red" elements (representing the vulnerable vertices) to an absolute minimum. We begin with a formal definition of the Red - Blue Set Cover Problem:

\begin{definition}[$RBSC$]
	Given a universe of elements partitioned into red $R$ and blue $B$, and a collection of subsets $\mathcal{S} = \{S_1, S_2, \dots, S_m\}$ with $S_i \subseteq R \cup B$, find a subcollection $\mathcal{S}' \subseteq \mathcal{S}$ such that $\mathcal{S}'$ is a set cover for $B$ and the number of covered red elements is minimum.
\end{definition}

A dominating set $D$ of a graph $G(V,E)$ is a subset of vertices $D \subseteq V$ such that every vertex not in $D$ is adjacent to at least one vertex in $D$. In $DSV$ the aim is to find a  Dominating Set for the non-vulnerable vertices which dominates the minimum number of vulnerable vertices. We formally define the Dominating Set with Vulnerable Vertices ($DSV$) problem:

\begin{definition}[$DSV$]
	Given a graph $G(V, E)$, where vertices are partitioned in non-vulnerable $V_N$ and vulnerable $V_L$, find a subset of vertices $S\subseteq V_N$ such that S is dominating set for $V_N$ and the number of dominated vulnerable vertices is minimum.
\end{definition}

The $DSV$ problem is at least as hard as Minimum Dominating Set ($Min-DS$). Indeed, every instance $G(V,E)$ of $Min-DS$ can be transformed to an instance of $DSV$ by adding a neighboring vulnerable vertex to each vertex of $G$, creating $G'$. The solution of $DSV$ on $G'$  gives the solution of $Min-DS$ for $G$ and vice versa.

\begin{algorithm}[H]
    \label{alg:DS_vul_nodes}
	\SetKwInOut{Input}{Input}
	\SetKwInOut{Output}{Output}
	\DontPrintSemicolon
	
	\Input{A graph $G(V, E)$, where vertices are partitioned in non-vulnerable $V_N$ and vulnerable $V_L$}
	\Output{A feasible solution $S$ for $DSV$}
	
	\vspace{0.5mm}
	
	$S_0 \leftarrow \emptyset$\;
    $R \leftarrow V_L$\;
    $B \leftarrow V_N$\;
	\ForEach{$\upsilon \in B$}{
		\uIf{$\sigma_L(\upsilon)=0$}{
			$S_0 \leftarrow S_0 \cup \{\upsilon\}$\;
            $B \leftarrow B \setminus \upsilon $\;
		}\Else{
			$X_\upsilon \leftarrow N[\upsilon]$\;
		}
	}
 $S_{RBSC} \leftarrow  RBSC(R,B, \bigcup_{\upsilon \in B}\{X_\upsilon\})$\;
 
	$S \leftarrow S_0 \cup S_{RBSC}$\;
	\textbf{return} $S$
	
	\caption{$DSV$}
	\label{alg:def3Algorithm}
\end{algorithm}

Algorithm~\ref{alg:DS_vul_nodes} reduces the $DSV$ problem to the Red-Blue Set Cover ($RBSC$) problem. Initially, the sets of vulnerable ($V_L$) and non-vulnerable ($V_N$) vertices are mapped to the red ($R$) and blue ($B$) elements of the $RBSC$ instance, respectively. We first select all non-vulnerable vertices with no vulnerable neighbors (i.e., $|\sigma_L(\upsilon)| = 0$) and add them to an initial solution set $S_0$. These vertices are then removed from the set of blue elements $B$ that need to be covered.

For each of the remaining vertices $\upsilon \in B$, we define a corresponding subset $X_\upsilon$ equal to its closed neighborhood $N[\upsilon]$. These subsets constitute the available sets in the $RBSC$ formulation, covering both blue and red elements. Finally, the $RBSC$ algorithm is called on the constructed instance $(R, B, \bigcup_{\upsilon \in B}\{X_\upsilon\})$ to compute a set cover $S_{RBSC}$. The algorithm returns the union $S_0 \cup S_{RBSC}$ as a feasible solution for the $DSV$ problem.

Let $\Delta_N = \max_{\upsilon \in V_N \wedge \sigma_L(\upsilon)>0}\{|N(\upsilon)|+1-|\sigma_L(\upsilon)|\}$ and
$\Delta_{L} = \max_{\upsilon \in V_L}|N(\upsilon)|$. By utilizing the reduction to the $RBSC$ problem, we can apply the theoretical known bounds to our formulation. The approximation of $RBSC$ fundamentally relies on the Weighted Set Cover ($WSC$) problem. While earlier results bounded the $WSC$ approximation by $1 + \ln(b)$, where $b$ is the maximum set size, the tighter bound of $H(b) - \frac{1}{2}$ established by \cite{duh1997approximation} provides a more accurate guarantee. Based on this, we derive the following theoretical results for the $DSV$ problem. 

\begin{theorem} Algorithm \ref{alg:def3Algorithm} is a
   $ \Delta_{L}\cdot(H(\Delta_{N})-\frac{1}{2}) $
    approximation for the $DSV$ problem, where $H$ is the harmonic function.
\end{theorem}

\begin{proof}
 The result follows from the approximation result of $WSC$ and Lemma 3.2 (\cite{PELEG200755}).
\end{proof}
\begin{theorem} Algorithm \ref{alg:def3Algorithm} is a
   $2\sqrt{|V|\cdot(H(\Delta_{N})-\frac{1}{2}})$
    approximation for the $DSV$ problem, where $H$ is the harmonic function.
\end{theorem}
\begin{proof}
 The result follows from the approximation result of $WSC$ and Theorem 3.5 (\cite{PELEG200755}).
\end{proof}

Since our reduction preserves the objective function exactly, all approximation bounds and hardness results for the Red-Blue Set Cover problem directly apply to the $DSV$ problem. 

\section{\texorpdfstring {Vertex Cover with Vulnerable Edges ($\mathit{VCVE}$)}{Vertex Cover with Vulnerable Edges (VCVE)}}
\label{sec:vcve} 
Having explored vulnerability constraints in the context of vertex domination, we now shift our focus to coverage problems by considering the case of vulnerable edges. The natural counterpart to the dominating set problems discussed earlier is the Vertex Cover with Vulnerable Edges ($VCVE$) problem (see Fig.~\ref{fig:vcve}). 

\begin{definition}[$\mathit{VCVE}$]
	Given a graph $G =(V, E)$ where edges are partitioned into non-vulnerable $E_N$ and vulnerable $E_L$, find a set of vertices $S\subseteq V$ covering all $E_N$ edges and the minimum number of $E_L$ edges. 
\end{definition}

The $\mathit{VCVE}$ problem is as hard as the Minimum Vertex Cover (Min-VC). Indeed, every instance $G=(V,E)$ of Min-VC can be transformed to an instance of $\mathit{VCVE}$ by adding a new vertex $\upsilon$ and a new vulnerable edge $(\upsilon,u)$ for each vertex $u$ of $G$. Let $G'$ be the created graph. Each solution for $\mathit{VCVE}$ on $G'$ directly yields a Min-VC solution for $G$, and vice versa.

To provide an approximation algorithm for $VCVE$, we draw inspiration from the analytical framework of the classical Red-Blue Set Cover ($RBSC$) problem. While the standard $RBSC$ analysis relies on a reduction to the Weighted Set Cover problem, our methodology strategically adapts this concept by establishing a reduction to the Weighted Vertex Cover problem. This tailored approach allows us to exploit the specific structural properties of vertex covers rather than generic sets. Consequently, this adaptation yields a significantly stronger theoretical guarantee: whereas the general $RBSC$ problem typically suffers from non-constant approximation bounds, our specific reduction enables us to achieve a strictly constant-factor approximation.

\begin{figure}
    	\centering
    	\includegraphics[scale=1.5]{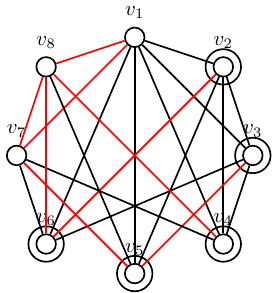}
    	\caption{A $VCVE$ instance. Red edges are vulnerable while the other ones are all non vulnerable. The circled vertices define a feasible solution $S=\{v_2, v_3, v_4, v_5, v_6\}$ with $|E_L(S)|=5$.}
    	\label{fig:vcve} 
    \end{figure}
    
  We adopt the following notation: For any vertex
$\upsilon \in V$, let $E(\upsilon)$ denote
the set of edges incident to $\upsilon$.
Furthermore, for a set of
vertices $S \subseteq V$,
$E_L(S)$ represents the set of
vulnerable edges covered by $S$, defined
as: $E_L(S)=\{(\upsilon,u)\in E_L: \upsilon\in S \lor u \in S\}$.\\

\textbf{The Algorithm.}
In Algorithm~\ref{alg:def4Algorithm}, the Weighted Vertex Cover ($WVC$) \cite{BARYEHUDA1981198} problem is used as a subroutine. In $WVC$, given an undirected graph
$G=(V,E)$ with a positive weight function $w: V \to \mathbb{R}^+$ assigning a positive weight $w(v)$ to each vertex $v \in V$, the goal is to find a subset $C\subseteq V$ such that every edge $e\in E$ is incident to at least one vertex $v\in C$ and the total weight of the selected vertices, denoted as $W(C)$ is minimized. In the algorithm, we first select all vertices with no adjacent vulnerable edges and add them to the solution. These vertices are then removed from the initial graph. Each remaining vertex is assigned a weight corresponding to the number of vulnerable edges incident to it. Vulnerable edges and isolated vertices are removed, and the resulting graph $G'$ is used in $WVC$ algorithm to provide a feasible solution.
The remaining graph $G'$ might be disconnected. We apply the $WVC$ algorithm for each of the components of $G'$.\\

\begin{algorithm}[H]
	\SetKwInOut{Input}{Input}
	\SetKwInOut{Output}{Output}
	\DontPrintSemicolon
	
	\Input{A graph $G=(V, E)$ where edges are partitioned into non-vulnerable $E_N$ and vulnerable $E_L$}
	\Output{A feasible solution $S$ for $\mathit{VCVE}$}

    $S \leftarrow \emptyset$,
	$S_0 \leftarrow \emptyset$,
	$V' \leftarrow V$,
        $E' \leftarrow E$\;
	\ForEach{$\upsilon \in V$}{
		
		\If{$E(\upsilon) \cap E_L \neq \emptyset$}{
			$w(\upsilon) \leftarrow |E(\upsilon)\cap E_L|$
		}\Else{
			$V' \leftarrow V' \setminus \{\upsilon\} $\\
			$E' \leftarrow E' \setminus E(\upsilon) $\\
			$S_0\leftarrow S_0 \cup \{\upsilon\}$
		}
	}
	
	$E' \leftarrow E' \setminus E_L $\\
	\ForEach{$\upsilon \in V'$}{
		\If{$E(\upsilon) = \emptyset$}{
			$V' \leftarrow V' \setminus \{\upsilon\} $
		}
	}
    Let $G'=(V',E')$\;
	$S_{WVC} \leftarrow WVC(G')$\;
 $S \leftarrow S_0 \cup S_{WVC}  $\;
	\textbf{return} $S$
	
	\caption{VCVE}
	\label{alg:def4Algorithm}
\end{algorithm}

\begin{lemma}\label{lemma:setCover}
For all $S \subseteq V$: $|E_L(S)|\leq W(S) \leq 2 \cdot|E_L(S)|$.
\end{lemma}
\begin{proof}
	By Algorithm~\ref{alg:def4Algorithm}, every weight of a vertex is equal to the number of vulnerable edges incident to it. So, every vulnerable edge adds one point to each neighboring vertex's weight. So:
    
	\[
	W(S)=\sum_{v \in S}w(v)=\sum_{e \in E_L(S)}\delta_S(e)
	\]
    
\noindent where, for $e=(\upsilon,u)$, $\delta_S(e) = 2$ when $\upsilon,u \in S$; $\delta_S(e) = 0$ when $\upsilon,u \notin S$; and $\delta_S(e) = 1$ otherwise.
By definition of $\delta_S(e)$, for every $e \in E_L(S)$, $1 \leq \delta_S(e) \leq 2$. So,
	\[
	\sum_{e \in E_L(S)}1\leq\sum_{e \in E_L(S)}\delta_S(e)\leq  \sum_{e \in E_L(S)}2
	\Rightarrow
	|E_L(S)|\leq W(S) \leq 2 \cdot|E_L(S)|
	\]
\end{proof}

\begin{theorem} Algorithm~\ref{alg:def4Algorithm} provides a $4$-approximation for $\mathit{VCVE}$.
\end{theorem}

\begin{proof}
	
	There exists an optimal solution $S^*$ such that $S_0 \subseteq S^*$, where $S_0$ the subset of vertices with no vulnerable incident edges.
	Let $G'=(V',E')$ be the graph that Algorithm~\ref{alg:def4Algorithm} defines and $S_{WVC}^*$ the optimal solution for $WVC$. By \cite{BARYEHUDA1981198,karakostas2009better} we obtain a feasible solution $S_{WVC}$ for $WVC$, such that $W(S_{WVC})\leq 2 \cdot W(S_{WVC}^*)$.
 
 By definition, $S^* \setminus S_0$ is a feasible solution for $WVC$ on $G'=(V',E')$ graph and so $W(S_{WVC}^*)\leq W(S^*\setminus S_0)$. We have:
  \[
  W(S_{WVC}) \leq 2 \cdot W(S^* \setminus S_0) \Rightarrow
  \]
\[
   W(S_{WVC}) + W(S_0) \leq 2 \cdot (W(S^* \setminus S_0) + W(S_0)) \Rightarrow
   \]
   \[
   W(S) \leq 2 \cdot W(S^*)
  \]

 By Lemma~\ref{lemma:setCover}, $|E_L(S)|\leq W(S)$ and $W(S^*)\leq 2\cdot |E_L(S^*)|$. So: 
	\[ |E_L(S)|\leq W(S) \leq 2 \cdot W(S^*) \leq 4 \cdot |E_L(S^*)|
  \]
\end{proof}

Although combinatorial reduction provides a constant-factor approximation and highlights the structural properties of the problem, a better approximation ratio can be achieved using linear programming. In the remainder of this section, we formulate the problem as an Integer Linear Program (ILP). We show that applying a rounding technique to its Linear Programming (LP) relaxation improves upon the previous approximation guarantee. Therefore, while the initial reduction provides a valid theoretical baseline, the LP rounding approach ultimately yields a tighter approximation bound. Without loss of generality, assume that a fixed ordering of the vertices is defined, so that each edge is counted exactly once. The ILP formulation for the Minimum $VCVE$ problem is defined as follows:

\setcounter{equation}{0}
\begin{align}
    \text{minimize}   \quad & \sum\limits_{\substack{i<j \\ (v_i, v_j) \in E_L}} z_{ij} \label{lp:objective} \\
    \text{subject to} \quad & x_i+x_j\geq 1, && \forall (v_i,v_j) \in E_N, \ i<j \label{lp:constraint1} \\
                            & z_{ij}\geq x_i \ , \ \ z_{ij}\geq x_j, && \forall (v_i, v_j) \in E_L, \ i<j\label{lp:constraint2}\\
                             & x_i \in \{0,1\}, && \forall v_i \in V \\
 & z_{ij}\in \{0,1\}, && \forall (v_i, v_j) \in E_L, \ i<j
\end{align}
\\The LP relaxation of the above formulation is obtained by replacing the integrality constraints (4) and (5) with the continuous relaxations $x_i \in [0,1]$ and $z_{ij}\in [0,1]$, respectively.

Let $(x^*,z^*)$ be the optimal of the ILP and $(\bar{x},\bar{z})$ be the optimal fractional solution of the LP relaxation. In Algorithm~\ref{alg:lp_rounding}, we construct an integer solution $(\hat{x},\hat{z})$ by using the following deterministic rounding scheme: a) for every vertex $v_i$, if $\bar{x}_i \geq 1/2$, we set $\hat{x}_i=1$; otherwise, we set $\hat{x}_i=0$, and b) for every edge $(v_i,v_j)$, if $\hat{x}_i=1$ or $\hat{x}_j=1$, we set $\hat{z}_{ij}=1$; otherwise, we set $\hat{z}_{ij}=0$.

\begin{lemma}\label{lemma:lp_feasible}
The rounding scheme provides a feasible solution for the $VCVE$ problem.
\end{lemma}
\begin{proof}
Let $(\bar{x},\bar{z})$ be the optimal fractional solution to the LP relaxation. According to constraint \eqref{lp:constraint1}, for every non-vulnerable edge $e = (u,v) \in E_N$, it holds that $\bar{x}_u + \bar{x}_v \geq 1$. This implies that at least one of the fractional variables, $\bar{x}_u$ or $\bar{x}_v$, must be greater than or equal to $1/2$. Consequently, the rounding scheme will set at least one of these variables to $1$. Thus, every non-vulnerable edge is covered by at least one selected vertex, ensuring the feasibility of the resulting integer solution.
\end{proof}

\begin{algorithm}[H]
	\SetKwInOut{Input}{Input}
	\SetKwInOut{Output}{Output}
	\DontPrintSemicolon

	\Input{A graph $G=(V, E)$ where edges are partitioned into non-vulnerable $E_N$ and vulnerable $E_L$}
	\Output{A feasible solution $S$ for $\mathit{VCVE}$}

	$(\bar{x}, \bar{z}) \leftarrow$ optimal fractional solution of LP relaxation\;
	$S \leftarrow \emptyset$\;
	\ForEach{$v_i \in V$}{
		\uIf{$\bar{x}_i \geq 1/2$}{
			$\hat{x}_i \leftarrow 1$\;
			$S \leftarrow S \cup \{v_i\}$\;
		}\Else{
			$\hat{x}_i \leftarrow 0$\;
		}
	}
	\ForEach{$(v_i, v_j) \in E_L$, $i < j$}{
		\uIf{$\hat{x}_i = 1$ \textbf{or} $\hat{x}_j = 1$}{
			$\hat{z}_{ij} \leftarrow 1$\;
		}\Else{
			$\hat{z}_{ij} \leftarrow 0$\;
		}
	}
	\textbf{return} $S$

	\caption{LP-Rounding for $\mathit{VCVE}$}
	\label{alg:lp_rounding}
\end{algorithm}

\begin{theorem}\label{thm:vcve_2approx}
The proposed Algorithm~\ref{alg:lp_rounding}  achieves a $2$-approximation for the $VCVE$.
\end{theorem}

\begin{proof}
Let $S^*$ be the optimal solution of the $VCVE$, $S$ the solution of Algorithm~\ref{alg:lp_rounding} , $(\hat{x},\hat{z})$ be the integer solution produced by the rounding scheme, and let $(\bar{x},\bar{z})$ be the optimal fractional solution of the LP relaxation. By Lemma \ref{lemma:lp_feasible}, we have established that $(\hat{x},\hat{z})$ is a feasible solution for the $VCVE$ problem, meaning it successfully covers all non-vulnerable edges $E_N$.

Let $OPT_{LP}$ denote the objective value of the optimal fractional solution and $OPT$ the objective value of the optimal integer solution. It holds $OPT_{LP} \le OPT=|E_L(S^*)|$. We have two cases:
\\\\
\begin{itemize}
    \item Let $\bar{z}_{ij} \geq 1/2$. It holds $\hat{z}_{ij} = 1 \leq 2\bar{z}_{ij}$.
    \item Let $\bar{z}_{ij} < 1/2$. It holds $\hat{z}_{ij} = 0 \leq 2\bar{z}_{ij}$.
\end{itemize}
We have:
\[
|E_L(S)|=\sum\limits_{\substack{i<j \\ (v_i, v_j) \in E_L}} \hat{z}_{ij}\leq \sum\limits_{\substack{i<j \\ (v_i, v_j) \in E_L}} 2\bar{z}_{ij} = 2\sum\limits_{\substack{i<j \\ (v_i, v_j) \in E_L}} \bar{z}_{ij} = 2OPT_{LP}\leq 2OPT=2 |E_L(S^*)|
\]
\end{proof}

\section{Conclusion}
\label{sec:conclusion} 
This work addresses the challenge of socially responsible influence maximization by introducing new variants of domination and coverage problems. We defined the $k$-Vertex Maximum Domination Ratio with Vulnerable Vertices problem and provided an approximation algorithm. We followed a path from the budgeted $k$-Vertex Maximum Domination Ratio with Vulnerable Vertices to the unbudgeted Vertex Maximum Domination Ratio with Vulnerable Vertices and then to Maximum Domination with Bounded number of Vulnerable Vertices problem, in order to achieve the approximation ratio $\frac{1}{\Delta+1}\cdot (1-e^{-\frac{1}{\Delta}}) \cdot (1-\frac{1}{e})\cdot \frac{k}{n}$. For the class of bounded degree graphs, this result further improves to $O(k/n)$. We also presented the Vertex Cover problem with vulnerable edges and the Dominating Set with vulnerable vertices, for which we  propose polynomial-time approximation algorithms.

Several questions remain open. For $k\textit{-}Max \ \mathit{DRVV}$, it is unclear whether a constant-factor approximation is achievable beyond specific graph classes or for small values of $k$. Tight inapproximability bounds for $\mathit{DVV}$ and $\mathit{DRVV}$ have yet to be established. Connected variants of $\mathit{DRVV}$ and $\mathit{DSV}$, as well as their behavior on structured graph families such as planar graphs or graphs of bounded treewidth, are also natural directions for future work. Finally, the general-graph bounds established in this paper may be loose for structured graph classes, and investigating whether stronger approximation ratios or even exact polynomial-time algorithms can be obtained for planar, bipartite, chordal, and bounded-treewidth graphs is a promising avenue.

\nocite{*}
\bibliographystyle{abbrvnat}
\bibliography{sample-dmtcs}
\label{sec:biblio}

\end{document}